\documentclass[
  aps,
  pra,
  reprint,
  longbibliography,
  nofootinbib,
  superscriptaddress
]{revtex4-2}

\usepackage[T1]{fontenc}
\usepackage[utf8]{inputenc}
\usepackage{amsmath,amssymb,amsfonts}
\usepackage{mathtools}
\usepackage{bm}
\usepackage{amsthm}
\usepackage{orcidlink}
\usepackage{xcolor}
\usepackage{tikz}
\usetikzlibrary{calc}
\usepackage{hyperref}

\hypersetup{
  colorlinks=true,
  linkcolor=blue,
  citecolor=blue,
  urlcolor=blue,
  pdftitle={Local Softmax and Global Weights in Non-Boolean Event Structures},
  pdfauthor={Karl Svozil},
  pdfkeywords={softmax, partition logic, no-disturbance, exotic probabilities, non-Boolean event structures}
}

\newtheorem{theorem}{Theorem}[section]
\newtheorem{proposition}[theorem]{Proposition}

\theoremstyle{definition}

\theoremstyle{remark}
\newtheorem{remark}[theorem]{Remark}

\begin{document}

\title{Local Softmax and Global Weights in Non-Boolean Event Structures}

\author{Karl Svozil\,\orcidlink{0000-0001-6554-2802}}
\email{karl.svozil@tuwien.ac.at}
\affiliation{Institute for Theoretical Physics, TU Wien,
Wiedner Hauptstrasse 8-10/136, A-1040 Vienna, Austria}

\date{\today}

\begin{abstract}
Softmax and related normalized response functions are widely used in
choice theory, machine learning, and cognitive science. In non-Boolean
event structures with overlapping contexts, however, local
normalization does not automatically yield a global probability
weight. We show that imposing single-valuedness on shared
atoms---equivalently, no-disturbance or consistent connectedness---
collapses generalized softmax rules to coordinate parametrizations of
the strictly positive part of the admissible-weight polytope. Any
strictly positive admissible weight can be represented in this way,
while boundary weights arise as limits. Exotic weights that exceed
classical or quantum bounds are therefore properties of the event
structure and the chosen weight, not of the normalizing link. The
resulting hierarchy separates local normalization, cross-context
gluing, Cauchy--Gleason linearity, and physical or cognitive
realizability.
\end{abstract}

\maketitle

\section{Introduction}
\label{sec:introduction}

The softmax rule
\begin{equation}
P_i=\frac{\exp(\beta u_i)}
{\sum_j \exp(\beta u_j)}
\label{eq:ordinary-softmax}
\end{equation}
is one of the most common ways of transforming scores, utilities,
energies, or response tendencies into probabilities. It appears in
statistical mechanics, discrete-choice theory, machine learning,
cognitive modeling, and social-science response models. Its appeal is
obvious: it is positive, normalized, differentiable, and controlled by
an inverse-temperature or discrimination parameter \(\beta\).

In an ordinary single choice set, Eq.~\eqref{eq:ordinary-softmax}
solves the immediate normalization problem. Given a finite list of
alternatives, their probabilities sum to one. In a non-Boolean event
structure, however, there is an additional compatibility problem.
There may be several overlapping contexts, and the same event may
occur in more than one of them. Normalizing separately inside each
context does not guarantee that shared events receive the same
probability in all contexts in which they occur.

This point is familiar in the language of quantum logics, test
spaces, exclusivity graphs, generalized urn models, automata, and
partition logics~\cite{greechie-1974,wright:pent,schaller-92,svozil-2001-eua,svozil-2018-b}.
A probability weight on such a
structure is not merely a family of normalized distributions, one for
each context. It is a single function on atoms that normalizes on
every context. If an atom is shared by two contexts, it has one value,
not two. This cross-context agreement will be called
\emph{single-valuedness on atoms}. It is the finite-test-space form
of \emph{no-disturbance}; in Bell-type scenarios the analogous
condition is no-signalling, and in the contextuality-by-default
literature it is closely related to consistent connectedness.

The purpose of this paper is to clarify the relation between this
condition and generalized softmax rules. We consider normalized
response functions of the form
\begin{equation}
P_{g,C}(a)=
\frac{g(u_C(a))}
{\sum_{b\in C}g(u_C(b))},
\label{eq:general-softmax-intro}
\end{equation}
where \(C\) is a context, \(a\in C\) is an atom, \(u_C(a)\) is a
score, and \(g\) is a strictly positive link function. The ordinary
softmax is obtained by \(g(x)=\exp(\beta x)\). Other choices of
\(g\) give logistic, probit, polynomial, power-law, or
\(q\)-exponential response models on suitable domains.

The central observation is elementary but important. If the scores
are context-dependent without further constraints,
Eq.~\eqref{eq:general-softmax-intro} gives only a family of
separately normalized response distributions. Such a family may
violate no-disturbance: the same shared atom may receive different
probabilities in different contexts. If single-valuedness or
no-disturbance is imposed, then the construction becomes exactly a
parametrization of an admissible weight on the logic. Conversely,
every strictly positive admissible weight can be written in
generalized softmax form, provided the link function has a suitable
positive range. Thus the softmax wrapper adds no new probability
theory once global consistency is required. Its explanatory force
begins only when the score functions are themselves constrained by
additional modelling assumptions.

This conclusion is useful because the admissible-weight polytope of a
non-Boolean event structure can be larger than the classical convex
hull of two-valued states. Some logics admit ``exotic'' weights:
normalized, single-valued assignments that are not classical mixtures
of dispersion-free states~\cite{greechie-1974,wright:pent}. In
odd-cycle scenarios, the admissible weight assigning \(1/2\) to each
cyclic atom and zero to the remaining atoms exceeds the classical
bound; for odd cycles of length at least five it also exceeds the
usual Lov{\'a}sz/Klyachko--Can--Binicio{\u g}lu--Shumovsky (KCBS)
quantum value~\cite{lovasz-79,Klyachko-2008,Bub-2009}. Such a weight
is obtained as a limiting softmax assignment. But the
beyond-classical or beyond-quantum character belongs to the weight's
location in the polytope, not to the exponential function itself.

There is also a nearby but distinct relation to Cauchy's functional
equation. Gleason-type theorems use additivity of probability
assignments on sufficiently rich event structures to force linear
Born-rule representations; Wright and Weigert have recently made the
connection between such Gleason-type arguments and Cauchy's
functional equation explicit~\cite{Wright2019}.
The present result does not derive Born-linearity. It concerns the
prior question of when locally normalized softmax responses glue into
a single admissible weight. Cauchy's equation can nevertheless enter
again at the level of the link function: if additive score
contributions are required to induce multiplicative unnormalized
weights, then the exponential link is singled out under the usual
regularity assumptions.

The paper is organized as follows. Section~\ref{sec:eventstructures}
defines finite pasted event structures and their weight polytopes.
Section~\ref{sec:generalized-softmax} introduces generalized softmax
rules, positive coordinates, and the distinction between local
normalization and global gluing. Section~\ref{sec:representation-theorem}
proves the representation theorem and records the gauge freedom and
boundary cases. Section~\ref{sec:cauchy} clarifies the relation to
Cauchy's functional equation and Gleason-type linearity theorems.
Section~\ref{sec:regions} relates the construction to classical,
quantum, and exotic regions. Section~\ref{sec:odd-cycles} treats odd
cycles and their half-weights. Section~\ref{sec:empirical} discusses
empirical and modelling implications. Section~\ref{sec:discussion}
collects broader consequences. Section~\ref{sec:conclusion}
concludes.

\section{Pasted event structures and admissible weights}
\label{sec:eventstructures}

We work with a finite event structure
\[
L=(A,\mathcal M),
\]
where \(A\) is a finite set of atoms and \(\mathcal M\) is a finite
family of contexts. Each context \(C\in\mathcal M\) is a subset of
\(A\), interpreted as a maximal set of mutually exclusive and jointly
exhaustive outcomes. In a partition-logic realization, the atoms are
subsets of a latent state space \(S\), and each context is a
partition of \(S\). In an abstract test-space or hypergraph
formulation, \(A\) and \(\mathcal M\) are taken as primitive.
The event structure is ``pasted'' in the quantum-logic sense:
different Boolean contexts are joined by identifying shared atoms,
in analogy with pasting constructions for orthomodular posets and
Boolean algebras~\cite{nav:91}.

We shall reserve the word ``gluing'' for the probabilistic
compatibility condition considered below: locally normalized
context-wise distributions glue if they assign the same probability
to every atom shared by two or more contexts, thereby defining a
single admissible weight on the already pasted event structure.

A probability weight on \(L\) is a function
\[
p:A\to[0,1]
\]
satisfying
\begin{equation}
\sum_{a\in C}p(a)=1
\qquad
\text{for all } C\in\mathcal M .
\label{eq:weight}
\end{equation}
The set of all such weights is a convex polytope:
\begin{equation}
\mathcal W(L)=
\left\{
p\in[0,1]^A:
\sum_{a\in C}p(a)=1
\text{ for all }C\in\mathcal M
\right\}.
\label{eq:weight-polytope}
\end{equation}
We call \(\mathcal W(L)\) the admissible-weight polytope.

An atom \(a\in A\) may occur in more than one context. In partition
logic such an atom is an intertwining atom: the same subset of the
latent state space appears as an atom of two or more partitions. The
definition~\eqref{eq:weight} assigns one number \(p(a)\) to that
atom. Hence single-valuedness on shared atoms is built into the
notion of a weight.

A two-valued weight, or dispersion-free state, is a weight
\[
v:A\to\{0,1\}
\]
that assigns value one to exactly one atom in every context. The
classical polytope associated with \(L\) is the convex hull of all
such two-valued weights:
\begin{equation}
\mathcal C(L)=
\operatorname{conv}\{v:v \text{ is a two-valued weight on }L\}.
\label{eq:classical-polytope}
\end{equation}
In a concrete partition-logic model built from a latent state space
\(S\), each latent state \(s\in S\) induces a point-evaluation state:
\[
v_s(B)=
\begin{cases}
1,& s\in B,\\
0,& s\notin B.
\end{cases}
\]
The convex hull of these intended point states gives the usual
classical latent-state model. Abstract logics may possess additional
two-valued states, but this distinction is not essential for the
softmax question considered here.

A quantum or Born-rule model is obtained when the atoms can be
represented by projectors \(\Pi_a\) on a Hilbert space such that
\[
\sum_{a\in C}\Pi_a=I
\qquad
\text{for all }C\in\mathcal M ,
\]
and probabilities are given by
\[
p(a)=\langle \psi|\Pi_a|\psi\rangle .
\]
The corresponding quantum set will be denoted schematically by
\(\mathcal Q(L)\). Its precise definition depends on dimension,
rank, and representation conventions. For exclusivity-graph
formulations, related relaxations are governed by Lov{\'a}sz-theta-type
quantities~\cite{lovasz-79,Cabello-2014-gtatqc}.

Whenever a Born-rule representation exists with projectors satisfying
context-wise completeness, the associated probabilities define an
admissible weight. Indeed,
\[
\sum_{a\in C}p(a)
=
\sum_{a\in C}\langle\psi|\Pi_a|\psi\rangle
=
\left\langle\psi\left|
\sum_{a\in C}\Pi_a
\right|\psi\right\rangle
=
\langle\psi|I|\psi\rangle
=
1.
\]
Moreover, because the same projector \(\Pi_a\) is assigned to the
same atom \(a\), the probability \(p(a)\) is single-valued across
contexts. Thus every such Born-rule model lies in \(\mathcal W(L)\).

The inclusion
\[
\mathcal Q(L)\subseteq\mathcal W(L)
\]
follows directly from blockwise completeness of the projectors.
If, in addition, one adopts the usual operational convention that
classical deterministic assignments are included among the allowed
quantum models, then one obtains the familiar chain
\[
\mathcal C(L)\subseteq\mathcal Q(L)\subseteq\mathcal W(L).
\]
The first inclusion is convention-dependent if one fixes a particular
Hilbert-space representation in advance: not every abstract
two-valued state need be represented inside that fixed representation.
Throughout, claims of being beyond quantum are therefore relative to a
specified quantum representation or to a specified graph-theoretic
quantum set, such as the Lov{\'a}sz-theta set for exclusivity graphs.

The terminology should also be fixed. In what follows,
single-valuedness or no-disturbance means membership in
\(\mathcal W(L)\). Classical noncontextuality, in the stronger sense
used in much of the contextuality literature, means membership in
\(\mathcal C(L)\). The distinction is essential: exotic weights are
single-valued and normalized on every context, but they need not be
classical.

\section{Generalized softmax and positive coordinates}
\label{sec:generalized-softmax}

Let \(g:I\to\mathbb R_{>0}\) be a strictly positive link function on
an interval \(I\subseteq\mathbb R\). Given a context \(C\in\mathcal M\)
and scores
\[
u_C:C\to I,
\]
define
\begin{equation}
P_{g,C}(a)=
\frac{g(u_C(a))}
{Z_C},
\qquad
Z_C=\sum_{b\in C}g(u_C(b)).
\label{eq:generalized-softmax}
\end{equation}
The ordinary softmax corresponds to \(g(x)=\exp(\beta x)\). The
choice \(g(x)=\Phi(\beta x)\), with \(\Phi\) the Gaussian cumulative
distribution function, gives a probit-type normalized response rule
on a suitable range. Logistic, polynomial, power-law, and
\(q\)-exponential links can be treated similarly, subject to
positivity and domain restrictions.

It is useful to factor the construction through the positive
coordinates
\begin{equation}
q_C(a)=g(u_C(a))>0 .
\label{eq:q-coordinate}
\end{equation}
Then Eq.~\eqref{eq:generalized-softmax} becomes simply
\begin{equation}
P_{g,C}(a)=
\frac{q_C(a)}
{\sum_{b\in C}q_C(b)}.
\label{eq:q-softmax}
\end{equation}
Thus the link function \(g\) does not itself carry probabilistic
content. It only specifies how positive coordinates \(q_C(a)\) are
represented by scores \(u_C(a)\). The substantive distinction is
whether the positive coordinates can be chosen so that the resulting
normalized probabilities glue consistently across contexts.

For each fixed context \(C\), Eq.~\eqref{eq:generalized-softmax}
defines an ordinary probability distribution:
\[
\sum_{a\in C}P_{g,C}(a)=1.
\]
It also respects additivity for unions of atoms inside that context.
Thus generalized softmax solves the \emph{local} normalization
problem.

It does not, by itself, solve the \emph{gluing} problem. Suppose
that \(a\) is an atom contained in two contexts \(C\) and \(C'\). In
general,
\[
P_{g,C}(a)
=
\frac{q_C(a)}
{\sum_{b\in C}q_C(b)}
\]
need not equal
\[
P_{g,C'}(a)
=
\frac{q_{C'}(a)}
{\sum_{b\in C'}q_{C'}(b)}.
\]
The same atom may be assigned two different probabilities. Such a
family is a legitimate collection of context-wise response
distributions, but it is not a probability weight on \(L\).

We therefore distinguish two notions.

\paragraph{Free context-wise generalized softmax.}
The scores \(u_C(a)\), or equivalently the positive coordinates
\(q_C(a)\), are allowed to depend freely on the context. The model
gives a normalized distribution for each \(C\), but no cross-context
consistency is assumed.

\paragraph{Single-valued generalized softmax.}
Whenever \(a\in C\cap C'\), one requires
\begin{equation}
P_{g,C}(a)=P_{g,C'}(a).
\label{eq:single-valued-softmax}
\end{equation}
Under this condition, the context-wise distributions glue into a
single function on atoms.

Condition~\eqref{eq:single-valued-softmax} is the no-disturbance or
consistent-connectedness condition for the present finite event
structure. It should not be confused with classical
noncontextuality. The former gives membership in \(\mathcal W(L)\);
the latter gives membership in \(\mathcal C(L)\).

The gluing condition can be stated directly in terms of positive
coordinates.

\begin{proposition}[Gluing condition for free scores]
\label{prop:gluing}
Let
\[
q_C(a)=g(u_C(a))>0,
\qquad
Z_C=\sum_{b\in C}q_C(b).
\]
The context-wise distributions
\[
P_C(a)=\frac{q_C(a)}{Z_C}
\]
glue to a single weight on atoms if and only if, for every shared
atom \(a\in C\cap C'\),
\begin{equation}
\frac{q_C(a)}{Z_C}
=
\frac{q_{C'}(a)}{Z_{C'}}.
\label{eq:gluing}
\end{equation}
Equivalently,
\begin{equation}
\frac{Z_C}{Z_{C'}}
=
\frac{q_C(a)}{q_{C'}(a)}
\label{eq:ratio}
\end{equation}
for every \(a\in C\cap C'\). In particular, if two contexts share
more than one atom, all shared atoms must induce the same normalizer
ratio. Around any cycle of overlapping contexts, the product of these
ratios must be \(1\).
\end{proposition}

\begin{proof}
The first statement is just the definition of single-valuedness of
the normalized probabilities. Rearranging Eq.~\eqref{eq:gluing}
gives Eq.~\eqref{eq:ratio}. If two contexts share several atoms, the
same ratio \(Z_C/Z_{C'}\) must be obtained from each of them.
Multiplying successive ratios around a closed chain of contexts
telescopes to \(1\).
\end{proof}

Proposition~\ref{prop:gluing} makes explicit why
context-independent scores are not sufficient. Even if
\(q_C(a)=q_{C'}(a)\) for a shared atom, the normalizing denominators
\(Z_C\) and \(Z_{C'}\) may differ. Then the shared atom receives
different probabilities.

\section{Representation theorem}
\label{sec:representation-theorem}

The basic relation between single-valued generalized softmax rules
and admissible weights is simple.

\begin{proposition}
\label{prop:softmax-to-weight}
Every single-valued generalized softmax family defines an admissible
weight on \(L\).
\end{proposition}

\begin{proof}
For each atom \(a\in A\), choose any context \(C\in\mathcal M\) such
that \(a\in C\), and set
\[
p(a)=P_{g,C}(a).
\]
Condition~\eqref{eq:single-valued-softmax} makes this definition
independent of the chosen context. For every \(C\in\mathcal M\),
\[
\sum_{a\in C}p(a)
=
\sum_{a\in C}P_{g,C}(a)
=1,
\]
by the normalization of Eq.~\eqref{eq:generalized-softmax}. Hence
\(p\in\mathcal W(L)\).
\end{proof}

The converse holds under mild assumptions on the link function.

\begin{theorem}[Representation theorem]
\label{thm:representation}
Let \(L=(A,\mathcal M)\) be finite. Let
\(g:I\to\mathbb R_{>0}\) be strictly monotone on an interval \(I\),
and suppose that its range contains an interval \((0,r)\) for some
\(r>0\). Then every strictly positive admissible weight
\(p\in\mathcal W(L)\) admits a single-valued generalized softmax
representation. That is, there exist scores \(u:A\to I\) such that
\begin{equation}
p(a)=
\frac{g(u(a))}
{\sum_{b\in C}g(u(b))}
\qquad
\text{for every } C\in\mathcal M \text{ and } a\in C.
\label{eq:softmax-representation}
\end{equation}
\end{theorem}

\begin{proof}
Since \(A\) is finite and \(p(a)>0\) for all \(a\), choose
\(\alpha>0\) so small that
\[
\alpha p(a)\in(0,r)
\qquad
\text{for all }a\in A.
\]
Because \(g\) is strictly monotone, it is injective and has a
well-defined inverse on its range. Define
\[
u(a)=g^{-1}(\alpha p(a)).
\]
Then
\[
g(u(a))=\alpha p(a).
\]
For any context \(C\),
\[
\sum_{b\in C}g(u(b))
=
\alpha\sum_{b\in C}p(b)
=
\alpha,
\]
because \(p\) is an admissible weight. Therefore, for \(a\in C\),
\[
\frac{g(u(a))}
{\sum_{b\in C}g(u(b))}
=
\frac{\alpha p(a)}{\alpha}
=
p(a).
\]
The resulting probabilities are single-valued on atoms.
\end{proof}

The construction is explicit: given \(p\) and \(g\), the scores are
\[
u(a)=g^{-1}(\alpha p(a))
\]
for any sufficiently small \(\alpha>0\) such that
\(\alpha p(a)\) lies in the range of \(g\) for every atom \(a\).

The theorem is deliberately deflationary. If the scores are completely
unrestricted, generalized softmax has no structural content beyond
positivity, local normalization, and the imposed gluing condition.
Its explanatory force begins only when the scores are constrained
independently of the target probabilities, for example by atom-only
scores, low-dimensional parametric forms, smoothness conditions across
contexts, or a mechanistic model of how utilities or evidential
strengths are generated.

\begin{remark}
The hypothesis that the range of \(g\) contains an interval
\((0,r)\) is only a convenient sufficient condition. For a fixed
weight \(p\), it is enough that the range of \(g\) contain a scaled
copy \(\alpha\{p(a):a\in A\}\) for some \(\alpha>0\). The interval
condition makes this automatic for all strictly positive weights on a
finite event structure.
\end{remark}

\begin{remark}[Gauge freedom]
\label{rem:gauge}
The generalized softmax representation is not unique. Multiplying all
positive coordinates \(q(a)=g(u(a))\) in a connected component of the
pasted logic by a common positive constant leaves all normalized
probabilities unchanged. For the exponential link \(g(u)=e^{\beta u}\),
this corresponds to adding a constant to all scores in that connected
component:
\[
u(a)\mapsto u(a)+c.
\]
For disconnected logics, this gauge freedom is componentwise. Thus
generalized softmax gives a many-to-one parametrization of the
relative interior of \(\mathcal W(L)\), not a unique coordinate chart.
\end{remark}

Boundary weights require a separate comment because many interesting
exotic weights assign zero probability to some atoms.

\begin{proposition}[Boundary weights]
\label{prop:boundary}
Let \(p\in\mathcal W(L)\). If \(p\) lies in the closure of the
strictly positive part of \(\mathcal W(L)\), and if \(0\) lies in the
closure of the range of \(g\), then \(p\) lies in the closure of the
generalized-softmax image.
\end{proposition}

\begin{proof}
Choose a sequence \(p_k\in\mathcal W(L)\) with \(p_k(a)>0\) for all
atoms \(a\) and \(p_k\to p\). By
Theorem~\ref{thm:representation}, each \(p_k\) has a generalized
softmax representation. Taking the limit gives \(p\). Atoms with
\(p(a)=0\) correspond to coordinates \(g(u_k(a))\to0\). For the
exponential link this means \(u_k(a)\to-\infty\).
\end{proof}

For the ordinary exponential softmax, exact zeros therefore require
extended scores or limiting procedures. Other response maps, such as
sparsemax or entmax, can produce exact zeros at finite scores, but
they are typically piecewise analytic rather than strictly positive
link normalizations of the form~\eqref{eq:generalized-softmax}.

The theorem has a deflationary consequence. Once single-valuedness is
imposed, generalized softmax is not an additional probability theory.
It is an analytic, many-to-one coordinate system on
\(\mathcal W(L)\), or on its positive part. The classification of a
weight as classical, quantum, or exotic is not determined by the link
function \(g\), but by the position of the resulting \(p\) inside or
outside the relevant subsets of \(\mathcal W(L)\).

\section{Relation to Cauchy's functional equation}
\label{sec:cauchy}

The preceding results should be distinguished from Cauchy-type
linearity theorems. The generalized-softmax construction concerns
local normalization and cross-context gluing. Once the probabilities
glue, one obtains an admissible weight
\[
p\in\mathcal W(L),
\]
that is, a single-valued function on atoms satisfying
\[
\sum_{a\in C}p(a)=1
\]
for every context \(C\). Additivity inside a Boolean block is then
imposed by defining the probability of a coarse-grained event as the
sum of the probabilities of its mutually exclusive atoms.

This is formally related to Cauchy's functional equation. In
Gleason-type theorems, one considers additive probability assignments
on richer structures, such as projection lattices or effect algebras.
The additivity condition has the Cauchy form
\begin{equation}
f(x+y)=f(x)+f(y),
\label{eq:cauchy-general}
\end{equation}
with \(x\) and \(y\) replaced by compatible effects or orthogonal
projections. Under positivity, boundedness, continuity, or
measurability assumptions, such additive functions are forced to be
linear, yielding the usual trace form of the Born rule. Wright and
Weigert~\cite{Wright2019} explicitly exploit this
connection to give an alternative proof of Busch's Gleason-type
theorem for effects~\cite{Busch-2003,caves-fuchs-2004}.

The softmax result proved here is different. It does not derive
linearity of a probability assignment, nor does it derive the Born
rule. It shows that, once single-valuedness or no-disturbance is
imposed, generalized softmax normalization is merely a parametrization
of an admissible weight. Whether that weight is classical,
Born-rule realizable, or exotic is a further geometric question about
its location in \(\mathcal W(L)\).

There is, however, another Cauchy-type route by which the ordinary
exponential softmax can be singled out from the class of generalized
links. Let \(q(u)>0\) be the unnormalized weight associated with a
score \(u\). If independent score contributions are required to add,
while their unnormalized weights multiply, then
\begin{equation}
q(u+v)=q(u)q(v).
\label{eq:multiplicative-link}
\end{equation}
This composition law is natural in settings where scores represent
additive evidence, utility increments, or energy-like contributions,
while the unnormalized weights represent independent multiplicative
factors. In statistical mechanics this is the familiar passage from
additive energies to Boltzmann factors; in response models it is the
assumption that independent score increments multiply odds-like
weights. It is not required by the gluing theorem.

Taking logarithms gives
\begin{equation}
\log q(u+v)=\log q(u)+\log q(v),
\label{eq:log-cauchy}
\end{equation}
which is Cauchy's equation.

\begin{proposition}[Exponential link from score composition]
\label{prop:exponential-link}
Let \(q:\mathbb R\to\mathbb R_{>0}\) satisfy
\[
q(u+v)=q(u)q(v)
\qquad
\text{for all }u,v\in\mathbb R.
\]
If \(q\) is measurable, or monotone, or bounded above on a measurable
set of positive measure, then there exists \(\beta\in\mathbb R\) such
that
\[
q(u)=e^{\beta u}.
\]
\end{proposition}

\begin{proof}
Define \(h(u)=\log q(u)\). Since \(q(u)>0\), this is well defined.
Equation~\eqref{eq:multiplicative-link} gives
\[
h(u+v)=h(u)+h(v).
\]
Under any of the stated regularity assumptions, \(h=\log q\) is a
regular solution of Cauchy's equation: measurability and monotonicity
are inherited from \(q\), while boundedness above of \(q\) on a
measurable set of positive measure implies boundedness above of
\(h\) there. Hence the standard regularity theorem for Cauchy's
equation gives~\cite{Aczel1966}
\[
h(u)=\beta u
\]
for some \(\beta\in\mathbb R\). Hence
\[
q(u)=e^{h(u)}=e^{\beta u}.
\]
\end{proof}

A complementary way to single out the exponential link is the
maximum-entropy characterization of softmax~\cite{Jaynes1957}. If a
context \(C\) is fixed and one maximizes Shannon entropy over
probability assignments \((p_a)_{a\in C}\), subject to normalization
and a linear constraint on the expected score,
\[
\sum_{a\in C}p_a=1,
\qquad
\sum_{a\in C}p_a u(a)=\bar u,
\]
then the resulting distribution has the ordinary exponential form
\[
p_a=
\frac{e^{\beta u(a)}}{\sum_{b\in C}e^{\beta u(b)}} .
\]
From this maximum-entropy perspective, the exponential softmax is the
canonical context-wise distribution compatible with a fixed expected
score. The gluing condition studied in the present paper then appears
as an additional cross-context constraint: when the same atom occurs
in two contexts, the two locally maximum-entropy distributions must
assign it the same probability. Thus the maximum-entropy derivation
selects the exponential form inside a single context, whereas
single-valuedness or no-disturbance determines whether the resulting
context-wise distributions assemble into a global weight on the
pasted logic. The Cauchy/composition argument above is dual in spirit:
instead of deriving exponentials from an entropy variational
principle, it derives them from multiplicativity of unnormalized
weights under additive score composition.

Thus Cauchy's equation can justify the exponential link under an
additional compositional axiom on scores. The gluing theorem itself,
however, does not depend on this axiom and applies equally to any
positive link function with suitable range.

The conceptual separation is therefore as follows. Cauchy/Gleason
arguments concern the additivity and linear extension of probability
weights on sufficiently rich additive domains. The present softmax
argument concerns the prior compatibility question of whether locally
normalized context-wise response probabilities define a single
admissible weight at all. Finally, the special status of the ordinary
exponential link arises only if one imposes a separate composition law
connecting additive scores to multiplicative unnormalized weights.

\section{Classical, quantum, and exotic regions}
\label{sec:regions}

The admissible-weight polytope \(\mathcal W(L)\) contains every
single-valued probability assignment compatible with the context-wise
normalization constraints. The classical polytope \(\mathcal C(L)\)
is generally smaller. It consists of convex mixtures of
dispersion-free assignments. In a partition model with a specified
latent state space \(S\), the intended classical model may be the
convex hull of the point-evaluation states induced by \(S\).

A weight
\[
p\in\mathcal W(L)\setminus\mathcal C(L)
\]
is nonclassical in the sense that it cannot be written as a classical
mixture of deterministic two-valued states. Such weights are often
called exotic in the literature on generalized urns and quantum
logics~\cite{wright:pent,greechie-1974}. They are normalized and
single-valued as weights on the logic, but they need not admit an
interpretation as ignorance over pre-existing values.

A further distinction concerns quantum realizability. A weight may
or may not be representable by a Born rule in a Hilbert-space model.
If a weight belongs to
\[
\mathcal W(L)\setminus\mathcal Q(L),
\]
then it is beyond the specified quantum set. This claim is
representation-dependent: one must specify the Hilbert-space
dimension, rank of projectors, and exclusivity convention. In
odd-cycle scenarios with the usual Lov{\'a}sz-theta convention, the
quantum maximum of the cyclic sum is given by
Eq.~\eqref{eq:lovasz-odd-cycle} below, whereas the admissible
half-weight gives \(n/2\).

Generalized softmax coordinates do not change these distinctions.
The same link function can represent a classical weight, a
quantum-realizable nonclassical weight, or an exotic beyond-quantum
weight, provided the target weight lies in the appropriate part of
\(\mathcal W(L)\). Thus one should distinguish carefully between:

\begin{enumerate}
\item \emph{Local normalization}: probabilities sum to one inside
each context.
\item \emph{Single-valuedness/no-disturbance}: shared atoms receive
the same probability across contexts.
\item \emph{Cauchy--Gleason linearity}: on richer domains,
Cauchy/Gleason-type assumptions may force additive probability
assignments to be linear.
\item \emph{Realizability}: the resulting weight is classical,
Born-rule realizable, exotic, or beyond the chosen quantum set.
\end{enumerate}

Softmax and generalized softmax rules address the first item
automatically. They address the second only if cross-context
constraints are imposed. They do not decide the third or fourth.

Thus admissible probabilities on a pasted logic can go beyond
classical, and in some cases beyond quantum, bounds. The odd-cycle
half-weights provide the simplest example: they are normalized and
single-valued on every context, hence belong to \(\mathcal W(L)\), but
they lie outside the classical convex hull; for odd cycles of length
at least five they also exceed the usual Lov{\'a}sz/KCBS quantum bound.
This beyond-classical or beyond-quantum character is not produced by
softmax normalization itself. It is a property of the event structure
and of the selected admissible weight. Softmax or generalized softmax
coordinates merely provide a parametrization---or, for boundary
weights such as the half-weight, a limiting parametrization---of that
weight.

\section{Odd cycles and exotic half-weights}
\label{sec:odd-cycles}

Odd cyclic logics provide the simplest examples in which admissible
weights exceed the classical convex hull. Even cycles admit an
alternating two-coloring of their cyclic atoms, and their independence
number is \(n/2\). The half-weight constructed below therefore
saturates but does not exceed the corresponding classical bound for
even \(n\), and no exotic behavior arises in the present sense. The
interesting case is odd \(n\).

Let \(n\geq3\) be odd. Consider the cyclic pasted logic with atoms
\[
a_1,\ldots,a_n,\qquad x_1,\ldots,x_n,
\]
and contexts
\begin{equation}
C_i=\{a_i,a_{i+1},x_i\},
\qquad i=1,\ldots,n,
\label{eq:odd-cycle-contexts}
\end{equation}
where indices are understood modulo \(n\). The atoms \(a_i\) are the
cyclic intertwining atoms; the atoms \(x_i\) are context-specific.

The assignment
\begin{equation}
p(a_i)=\frac12,\qquad p(x_i)=0
\quad (i=1,\ldots,n)
\label{eq:odd-cycle-half-weight}
\end{equation}
is an admissible weight, since
\[
p(a_i)+p(a_{i+1})+p(x_i)
=
\frac12+\frac12+0
=
1
\]
for every context \(C_i\).

It is not a classical mixture of two-valued states. In any
dispersion-free assignment, no two adjacent cyclic atoms can both
receive value one. Hence the number of cyclic atoms assigned value
one is at most the independence number of the odd cycle,
\[
\alpha(C_n)=\frac{n-1}{2}.
\]
Therefore every classical mixture satisfies
\begin{equation}
\sum_{i=1}^n p(a_i)\leq \frac{n-1}{2}.
\label{eq:odd-cycle-classical-bound}
\end{equation}
The half-weight~\eqref{eq:odd-cycle-half-weight} gives
\[
\sum_{i=1}^n p(a_i)=\frac n2,
\]
and therefore violates the classical bound by \(1/2\).

For the usual exclusivity-graph quantum model on an odd cycle, the
Lov{\'a}sz-theta value is~\cite{lovasz-79}
\begin{equation}
\vartheta(C_n)=
\frac{n\cos(\pi/n)}
{1+\cos(\pi/n)}.
\label{eq:lovasz-odd-cycle}
\end{equation}
For odd \(n\geq5\),
\[
\frac n2 >
\frac{n\cos(\pi/n)}
{1+\cos(\pi/n)}.
\]
Indeed, since \(0<\pi/n<\pi/2\), one has
\(\cos(\pi/n)<1\), and hence
\[
1+\cos(\pi/n)>2\cos(\pi/n),
\]
which is exactly the displayed inequality after multiplication by
\(n/[2(1+\cos(\pi/n))]\). Thus the half-weight also lies beyond the
usual quantum set for these odd-cycle exclusivity scenarios. For
\(n=5\), Eq.~\eqref{eq:lovasz-odd-cycle} gives \(\sqrt5\), recovering
the KCBS value.

\paragraph{Generalized-softmax path between midpoint states and half-weights.}
The softmax construction can be stated in a way that makes the
inessential role of the exponential link explicit. Let \(g\) be any
strictly positive link function and choose a symmetric global score
assignment
\[
u(a_i)=u_0,\qquad u(x_i)=t
\quad (i=1,\ldots,n).
\]
Then each context \(C_i=\{a_i,a_{i+1},x_i\}\) has normalizer
\[
Z_i=2g(u_0)+g(t),
\]
and therefore
\begin{align}
p_t(a_i)&=p_t(a_{i+1})
=
\frac{g(u_0)}{2g(u_0)+g(t)}, \nonumber
\\
p_t(x_i)&=
\frac{g(t)}{2g(u_0)+g(t)}.  \nonumber
\end{align}
Equivalently, writing
\[
r(t)=\frac{g(t)}{g(u_0)}
\]
gives the link-independent form
\[
p_t(a_i)=p_t(a_{i+1})
=
\frac{1}{2+r(t)},
\qquad
p_t(x_i)=
\frac{r(t)}{2+r(t)}.
\]
This family is single-valued on shared atoms: every cyclic atom
\(a_i\) has the same probability in the two contexts in which it
occurs.

The parameter \(r(t)\) interpolates between three instructive
regimes. For \(r(t)\to\infty\), one obtains the deterministic
midpoint assignment
\[
p_t(x_i)\to1,\qquad p_t(a_i)\to0.
\]
This is a valid two-valued weight on the pasted logic. What fails on
an odd cycle is a stronger symmetric coloring picture in which, after
fixing all midpoint atoms, the two remaining labels alternate
consistently around the cycle. For \(r(t)=1\), all three atoms in each
context have equal positive coordinate and
\[
p_t(a_i)=p_t(a_{i+1})=p_t(x_i)=\frac13.
\]
For \(r(t)\to0\), the midpoint atoms are suppressed and one obtains
the exotic boundary half-weight
\[
p_t(a_i)\to\frac12,\qquad p_t(x_i)\to0,
\]
which is Eq.~\eqref{eq:odd-cycle-half-weight}. Thus generalized
softmax supplies a smooth fractional path through the same odd-cycle
obstruction: the failed integral alternation of colors is replaced,
at the opposite boundary, by the symmetric fractional assignment
\(1/2,1/2,0\) in every context.

For the ordinary exponential link \(g(u)=e^{\beta u}\), one may set
\(u_0=0\). Then
\[
r(t)=e^{\beta t},
\]
so the three regimes correspond to \(t\to+\infty\), \(t=0\), and
\(t\to-\infty\), respectively, assuming \(\beta>0\). Equivalently,
setting \(t=-M\) and sending \(M\to\infty\) gives the limiting
construction
\[
p_t(a_i)\to\frac12,\qquad p_t(x_i)\to0.
\]
For finite real-valued scores the exponential softmax remains
strictly positive and therefore approximates, but does not exactly
attain, this boundary weight. For a general positive link \(g\), the
half-weight is approached whenever
\[
r(t)=\frac{g(t)}{g(u_0)}\to0,
\]
or, equivalently, whenever the positive coordinate assigned to the
context-specific atom \(x_i\) is suppressed relative to the cyclic
atoms. This requires \(0\) to lie in the closure of the attainable
coordinate ratios.

Table~\ref{tab:pentagon-three-regimes} records the explicit
probability values for the three distinguished assignments in the
pentagon case, and Fig.~\ref{fig:pentagon-ax} shows the corresponding
logic in the notation used here.

\begin{table}[ht]
\caption{Explicit probability values for the ten atoms of the pentagon
logic in the three distinguished regimes of the path: the midpoint
limit \(r\to\infty\), the uniform state \(r=1\), and the half-weight
limit \(r\to0\), where \(r=g(t)/g(u_0)\). For the ordinary
exponential parametrization with \(u_0=0\) and \(\beta>0\), these
correspond to \(t\to+\infty\), \(t=0\), and \(t\to-\infty\),
respectively.}
\label{tab:pentagon-three-regimes}
\begin{ruledtabular}
\begin{tabular}{lcccccccccc}
Regime
& \(a_1\) & \(x_1\)
& \(a_2\) & \(x_2\)
& \(a_3\) & \(x_3\)
& \(a_4\) & \(x_4\)
& \(a_5\) & \(x_5\) \\
\colrule
Midpoint limit
& \(0\) & \(1\)
& \(0\) & \(1\)
& \(0\) & \(1\)
& \(0\) & \(1\)
& \(0\) & \(1\) \\
Uniform state
& \(\tfrac13\) & \(\tfrac13\)
& \(\tfrac13\) & \(\tfrac13\)
& \(\tfrac13\) & \(\tfrac13\)
& \(\tfrac13\) & \(\tfrac13\)
& \(\tfrac13\) & \(\tfrac13\) \\
Half-weight
& \(\tfrac12\) & \(0\)
& \(\tfrac12\) & \(0\)
& \(\tfrac12\) & \(0\)
& \(\tfrac12\) & \(0\)
& \(\tfrac12\) & \(0\)
\end{tabular}
\end{ruledtabular}
\end{table}

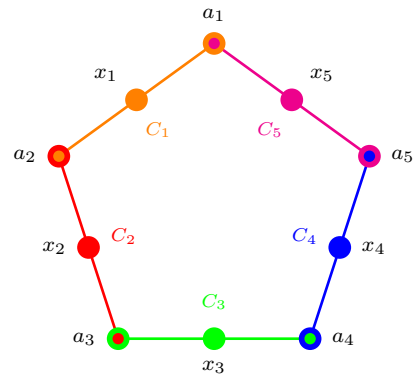
\begin{figure}[ht]
\centering
\begin{tikzpicture}[scale=0.36]

\newdimen\ms
\ms=0.05cm

\tikzstyle{every path}=[line width=1pt]
\tikzstyle{c3}=[circle,inner sep={\ms/8},minimum size=6*\ms]
\tikzstyle{c2}=[circle,inner sep={\ms/8},minimum size=3*\ms]

\newdimen\R
\R=6cm

\path
  ({90 + 0 * 360 /5}:\R) coordinate(a1)
  ({90 + 36 + 0 * 360 /5}:{\R * sqrt((25+10*sqrt(5))/(50+10*sqrt(5)))}) coordinate(x1)
  ({90 + 1 * 360 /5}:\R) coordinate(a2)
  ({90 + 36 + 1 * 360 /5}:{\R * sqrt((25+10*sqrt(5))/(50+10*sqrt(5)))}) coordinate(x2)
  ({90 + 2 * 360 /5}:\R) coordinate(a3)
  ({90 + 36 + 2 * 360 /5}:{\R * sqrt((25+10*sqrt(5))/(50+10*sqrt(5)))}) coordinate(x3)
  ({90 + 3 * 360 /5}:\R) coordinate(a4)
  ({90 + 36 + 3 * 360 /5}:{\R * sqrt((25+10*sqrt(5))/(50+10*sqrt(5)))}) coordinate(x4)
  ({90 + 4 * 360 /5}:\R) coordinate(a5)
  ({90 + 36 + 4 * 360 /5}:{\R * sqrt((25+10*sqrt(5))/(50+10*sqrt(5)))}) coordinate(x5)
;

\draw[color=orange]  (a1) -- (x1) -- (a2);
\draw[color=red]     (a2) -- (x2) -- (a3);
\draw[color=green]   (a3) -- (x3) -- (a4);
\draw[color=blue]    (a4) -- (x4) -- (a5);
\draw[color=magenta] (a5) -- (x5) -- (a1);

\draw (a1) coordinate[c3,fill=orange,label=90:{\footnotesize \(a_1\)}];
\draw (a1) coordinate[c2,fill=magenta];

\draw (a2) coordinate[c3,fill=red,label={left:\footnotesize \(a_2\)}];
\draw (a2) coordinate[c2,fill=orange];

\draw (a3) coordinate[c3,fill=green,label={left:\footnotesize \(a_3\)}];
\draw (a3) coordinate[c2,fill=red];

\draw (a4) coordinate[c3,fill=blue,label={right:\footnotesize \(a_4\)}];
\draw (a4) coordinate[c2,fill=green];

\draw (a5) coordinate[c3,fill=magenta,label={right:\footnotesize \(a_5\)}];
\draw (a5) coordinate[c2,fill=blue];

\draw (x1) coordinate[c3,fill=orange,label={above left:\footnotesize \(x_1\)}];
\draw (x2) coordinate[c3,fill=red,label={left:\footnotesize \(x_2\)}];
\draw (x3) coordinate[c3,fill=green,label={below:\footnotesize \(x_3\)}];
\draw (x4) coordinate[c3,fill=blue,label={right:\footnotesize \(x_4\)}];
\draw (x5) coordinate[c3,fill=magenta,label={above right:\footnotesize \(x_5\)}];

\node[orange,font=\scriptsize]  at ($(x1)!0.28!(0,0)$) {\(C_1\)};
\node[red,font=\scriptsize]     at ($(x2)!0.28!(0,0)$) {\(C_2\)};
\node[green,font=\scriptsize]   at ($(x3)!0.28!(0,0)$) {\(C_3\)};
\node[blue,font=\scriptsize]    at ($(x4)!0.28!(0,0)$) {\(C_4\)};
\node[magenta,font=\scriptsize] at ($(x5)!0.28!(0,0)$) {\(C_5\)};

\end{tikzpicture}
\caption{Pentagon logic in the notation of
Sec.~\ref{sec:odd-cycles}. The cyclic atoms
\(a_1,\ldots,a_5\) are intertwining atoms, each shared by two adjacent
contexts. The atoms \(x_1,\ldots,x_5\) are context-specific. Each
colored broken line represents one context
\(C_i=\{a_i,a_{i+1},x_i\}\), with indices understood modulo~\(5\).}
\label{fig:pentagon-ax}
\end{figure}

It is worth noting, however, that for every odd \(n\ge3\) the
two-valued state
\[
p(x_i)=1,\qquad p(a_i)=0\quad (i=1,\ldots,n)
\]
is classical as a dispersion-free weight but ``exotic'' in a
different, coloring-theoretic sense: it cannot be refined to a
consistent three-coloring of the odd-cycle contexts. Indeed, fixing
all midpoint atoms \(x_i\) to one color would force the two remaining
colors to alternate on the cyclic atoms \(a_i\) around an odd cycle,
which is impossible~\cite{svozil-2025-color}. Thus the
coloring-theoretic obstruction already appears for the triangle
\(n=3\). The pentagon shown in Fig.~\ref{fig:pentagon-ax} and
Table~\ref{tab:pentagon-three-regimes} is the first odd-cycle case
where the classical, Lov{\'a}sz/KCBS, and admissible half-weight
bounds are all distinct.

\subsection{Triangle}

For \(n=3\), the construction gives Wright's triangle logic. The
half-weight has
\[
\sum_{i=1}^3p(a_i)=\frac32,
\]
whereas the classical bound is
\[
\alpha(C_3)=1.
\]
Thus the triangle already exhibits the basic phenomenon: an admissible
single-valued weight can lie outside the classical convex hull. The
graph-theoretic quantum comparison at \(n=3\), however, is degenerate.
The three-cycle coincides with the complete graph \(K_3\), and it is
not a KCBS-type scenario in the usual sense. We therefore use the
triangle only as the minimal example of an admissible but nonclassical
weight. The pentagon \(n=5\), shown in Fig.~\ref{fig:pentagon-ax}, is
the first nontrivial case in which the classical bound, the quantum
Lov{\'a}sz/KCBS bound, and the admissible half-weight are all distinct.

\subsection{Pentagon}

For \(n=5\), the construction gives the pentagon/pentagram logic of
Fig.~\ref{fig:pentagon-ax}. The ten atoms are
\[
a_1,\ldots,a_5,\qquad x_1,\ldots,x_5,
\]
with contexts
\[
C_i=\{a_i,a_{i+1},x_i\},
\]
where indices are understood modulo~\(5\). The three distinguished
assignments along the softmax path are listed explicitly in
Table~\ref{tab:pentagon-three-regimes}.

For the half-weight,
\[
p(a_i)=\frac12,\qquad p(x_i)=0,
\]
one obtains
\[
\sum_{i=1}^5p(a_i)=\frac52.
\]
This exceeds the classical bound
\[
\alpha(C_5)=2
\]
and also the usual quantum/KCBS bound
\[
\vartheta(C_5)=\sqrt5.
\]
Thus the pentagon half-weight is admissible, nonclassical, and beyond
the standard quantum five-cycle model.

Along the generalized-softmax path, the cyclic weight is
\[
\sum_{i=1}^5 p_t(a_i)
=
\frac{5}{2+r(t)},
\qquad
r(t)=\frac{g(t)}{g(u_0)}.
\]
It exceeds the classical bound \(2\) precisely when
\[
\frac{5}{2+r(t)}>2,
\]
that is,
\[
r(t)<\frac12.
\]
It exceeds the usual KCBS/Lov{\'a}sz value \(\sqrt5\) precisely when
\[
\frac{5}{2+r(t)}>\sqrt5,
\]
or equivalently
\[
r(t)<\sqrt5-2.
\]
Thus finite generalized-softmax coordinates already enter the
nonclassical and beyond-quantum regions once the midpoint atoms
\(x_i\) are sufficiently suppressed relative to the cyclic atoms
\(a_i\). The exact half-weight \(5/2\), however, is reached only in
the boundary limit \(r(t)\to0\). For the ordinary exponential link,
this condition is \(e^{\beta t}\to0\), i.e. \(t\to-\infty\) for
\(\beta>0\).

The softmax representation therefore does not explain the
beyond-quantum character of the weight. It only parametrizes, in the
closure of the positive generalized-softmax image, a point already
present in the admissible-weight polytope \(\mathcal W(L)\).

\section{Empirical and modelling implications}
\label{sec:empirical}

In cognitive and social-science applications, a score
\(u_C(a)\) may represent utility, salience, evidential strength,
diagnostic support, response tendency, or contextual attractiveness.
Equation~\eqref{eq:generalized-softmax} then turns such scores into
a probability distribution over the outcomes available in context
\(C\).

The preceding analysis shows that this step should not be confused
with a global probabilistic model. A context-wise softmax may fit
response frequencies in each context separately while failing to
assign a single probability to atoms shared by multiple contexts. In
such a case the model describes contextual response formation, not a
single-valued weight on a pasted logic.

For empirical applications, the analysis suggests a three-stage
procedure. First, estimate the response probabilities separately in
each context. Second, test whether shared atoms have equal
probabilities across the contexts in which they occur, within
sampling error. In a between-subjects design, these estimates come
from different samples, so a failure of equality may reflect genuine
disturbance, selection effects, or a mis-specified exclusivity
structure rather than a single identifiable mechanism. Third, if the
single-valuedness condition holds to the required statistical
precision, reconstruct the global weight \(p\in\mathcal W(L)\) and
test its location relative to the classical polytope
\(\mathcal C(L)\) and, where appropriate, a specified quantum set
\(\mathcal Q(L)\). Membership in \(\mathcal C(L)\) is a
linear-programming question; for small logics it can also be tested
directly by enumerating the two-valued states. Membership in common
graph-theoretic relaxations of \(\mathcal Q(L)\) is typically treated
by semidefinite programming or by known bounds such as
Eq.~\eqref{eq:lovasz-odd-cycle}.

Generalized softmax coordinates may be useful for parametrization and
estimation, but they do not replace these structural tests. A
successful softmax fit inside each context does not tell us whether
the observed probabilities are classical, quantum-like, exotic, or
not even globally compatible. The decisive issue is not local
normalization but the geometry of the global weight.

The representation theorem also shows where genuine modelling
assumptions enter. If the scores are completely unrestricted, then a
generalized softmax representation has little explanatory content:
any strictly positive admissible weight can be represented. A
substantive softmax model must therefore constrain the scores
independently of the target probabilities. Examples include atom-only
scores \(u_C(a)=u(a)\), low-dimensional parametric forms
\[
u_C(a)=\bm\theta\cdot \bm\phi(a,C),
\]
regularity conditions across contexts, or mechanistic assumptions
about how utilities, saliences, or evidential strengths are formed.
Only such restrictions turn generalized softmax from a
reparametrization into an explanatory model.

This hierarchy is particularly relevant to quantum cognition
~\cite{Busemeyer2012,Pothos2013}. Context-wise response models are
often descriptively flexible, but flexibility is not the same as a
structural explanation. The relevant structural questions are whether
the fitted probabilities glue, whether the glued weight is classical,
whether it is Born-rule realizable, whether a Cauchy/Gleason-type
linearity theorem applies on a richer domain, and whether the weight
is exotic relative to the chosen event structure.

\section{Discussion}
\label{sec:discussion}

The main result can be summarized as follows:
\[
\begin{aligned}
&\text{single-valued generalized softmax} \\
&\qquad =
\text{admissible weight in link-function coordinates}.
\end{aligned}
\]
This equivalence is mathematically elementary, but it prevents a
common conflation. Softmax normalization is sometimes treated as a
probabilistic principle. In a non-Boolean event structure, however,
it is only a local normalization scheme unless cross-context gluing
conditions are added. Once those conditions are added, softmax no
longer defines a new model class; it simply parametrizes the
admissible weights that were already present.

This observation also clarifies the role of link functions. The
exponential is not essential for the gluing theorem. Any positive
link with a suitable range can be used to coordinatize the positive
part of the weight polytope. Changing the link changes the score
scale, response sensitivity, and statistical convenience, but not the
underlying distinction between classical, quantum, and exotic weights.
The exponential link becomes distinguished only under an additional
Cauchy-type composition principle: additive score contributions must
correspond to multiplicative unnormalized weights.

The exotic examples are especially instructive. The odd-cycle
half-weights are single-valued in the operational sense: each atom
has one value and each context is normalized. They are nevertheless
not classical mixtures of two-valued states. For odd cycles of length
at least five, the same assignments exceed the usual Lov{\'a}sz/KCBS
quantum bounds. A limiting softmax can represent these weights, but
it does not explain them. Their existence is a property of the
pasted event structure.

One should therefore be cautious with phrases such as ``softmax
beyond quantum.'' The more precise statement is: generalized softmax
coordinates can represent admissible weights that are beyond the
classical or quantum sets. Boundary exotic weights with zero
components are represented only as limits or, equivalently, as points
in the closure of the positive generalized-softmax image. The
beyondness is not due to the normalizing function; it is due to the
selected point in \(\mathcal W(L)\).

This distinction echoes the classic Bell--Peres intuition regarding
contextuality. If one insists on assigning outcomes to unperformed
measurements, Bell-type violations force the value assigned to an
observable to depend on the context in which it is considered. Peres
expressed this by replacing a putative setting-independent result
\(r_{j\alpha}\) by an explicitly context-dependent result
\(r_{j\alpha}(\beta)\), or more generally
\(r_{j\alpha}(\beta,\gamma,\ldots)\), where the additional arguments
refer to the remaining measurement context~\cite[p.~747]{peres222}.
Similarly, in an explicitly contextual
Clauser--Horne--Shimony--Holt (CHSH) notation one may write
\(E(x,y)\) as \(E(x_y,y_x)\), with contextuality manifesting itself
through \(x_y\ne x_{y'}\) when \(y\ne y'\)
~\cite{svozil-2011-enough}.

In the present framework, a free context-wise softmax without the
gluing condition is a probabilistic analogue of this explicit
context-dependence. A single nominal atom \(a\) is replaced by
context-indexed copies \(a_C,a_{C'},\ldots\), and the probabilities
assigned to these copies need not agree. In contrast, an admissible
weight on a pasted event structure imposes the strict identification
of shared atoms: if \(a\in C\cap C'\), then \(a\) carries one
probability \(p(a)\), not two.

This distinction changes the mathematical problem. In a strictly
pasted event structure, shared atoms are the points at which Boolean
blocks are glued together. They are responsible for the global
constraints that make Bell-, KCBS-, and Kochen--Specker-type
obstructions possible. If the shared atoms are split into
context-indexed copies, the original pasted structure is replaced by
a disjoint family of context-wise Boolean algebras. On that disjoint
structure alone there is no logical obstruction to forming a global
joint distribution over the context-indexed variables; one may, for
example, take a product of the context-wise distributions. What has
been removed is not contextuality as an empirical phenomenon, but the
strict logical identification of atoms across contexts.

The Contextuality-by-Default (CbD) framework~\cite{Dzhafarov-19PhysRevA,Dzhafarov2026} starts precisely from this
context-indexed description. Random variables measuring the same
content in different contexts are not identified at the outset; they
are treated as distinct variables belonging to different bunches.
Empirical disturbance, or inconsistent connectedness, is therefore
built into the formalism. CbD then reintroduces a probabilistic
analogue of gluing by asking whether these context-indexed variables
can be coupled so that same-content variables agree as often as their
individual distributions allow. A system is contextual, in the CbD
sense, if no global coupling can satisfy all these maximal-connection
requirements simultaneously. This residual system-wide obstruction is
what CbD calls covert context-dependence~\cite{Dzhafarov2026}.

Thus the two approaches address different regimes. CbD provides a
probabilistic methodology for systems in which context-indexing and
empirical disturbance are taken as primitive. The present paper
instead isolates the strictly glued regime of pasted event
structures. In that regime, shared atoms are identified from the
start, and an admissible weight is a single-valued function on those
atoms. The main result then says that, once this strict
single-valuedness is imposed, generalized softmax normalization
contributes no new probabilistic mechanics: it is only a coordinate
parametrization of the admissible-weight polytope \(\mathcal W(L)\).
Whether the resulting weight is classical, quantum-realizable,
exotic, or beyond a specified quantum set is determined by the
geometry of the pasted event structure, not by the softmax link.

A natural next step is not to ask whether generalized
softmax can represent exotic weights---it can, at least in the
closure---but which constrained score families can do so. Atom-only
scores \(u_C(a)=u(a)\), low-dimensional forms
\(u_C(a)=\bm\theta\cdot\bm\phi(a,C)\), locality restrictions on
\(\bm\phi\), or smoothness assumptions across neighboring contexts
may cut out proper subregions of \(\mathcal W(L)\). Characterizing
these subregions, and comparing them with \(\mathcal C(L)\),
\(\mathcal Q(L)\), and the exotic part of \(\mathcal W(L)\), is the
substantive modelling problem left open by the present deflationary
result.

Several extensions suggest themselves. First, one may study sparse or
thresholded response maps that produce zero weights at finite scores.
Second, one may impose additional locality, dimensionality, or
regularity constraints on the score function \(u\), thereby obtaining
proper subfamilies of \(\mathcal W(L)\). Third, one may investigate
statistical estimation procedures in which generalized softmax
coordinates are used for constrained inference over \(\mathcal C(L)\),
\(\mathcal Q(L)\), or \(\mathcal W(L)\). Fourth, one may ask when
Cauchy/Gleason-type extension theorems apply to weights obtained from
empirical response models, thereby connecting finite pasted logics
with richer effect-algebraic or Hilbert-space structures. These
directions would turn the present geometric observation into a
practical modeling framework.

\section{Conclusion}
\label{sec:conclusion}

Generalized softmax rules transform scores into normalized
probabilities within a context. In non-Boolean event structures with
overlapping contexts, this is not enough to define a probability
weight. Shared atoms must receive identical probabilities wherever
they occur. If this no-disturbance or single-valuedness condition is
not imposed, one has a family of context-wise response distributions,
not a weight on the pasted logic. If it is imposed, the generalized
softmax rule collapses to a coordinate representation of an
admissible weight.

Under mild assumptions on the positive link function, every strictly
positive admissible weight can be represented in this way. Boundary
weights, including the exotic half-weights on odd cycles, are
obtained as limits when zero lies in the closure of the link's range.
Therefore single-valued generalized softmax is not a new probability
theory. It is a parametrization of \(\mathcal W(L)\). The interesting
distinctions are instead geometric: whether the resulting weight is
classical, quantum-realizable, exotic, or beyond the specified
quantum set.

This result is adjacent to, but distinct from, Cauchy/Gleason
linearity theorems. Cauchy-type additivity arguments can force
linear/Born-rule probability assignments on sufficiently rich
additive domains. The present theorem instead concerns the gluing of
locally normalized response probabilities into a single admissible
weight. Cauchy's equation reappears at the level of the link function
only if one imposes the additional axiom that additive scores give
multiplicative unnormalized weights, in which case the ordinary
exponential softmax is singled out.

Odd cyclic logics provide minimal examples. They support half-weights
that are normalized and single-valued but not classical; for odd
cycles of length at least five, these weights exceed the usual
Lov{\'a}sz/KCBS quantum bound. They are obtained as limiting softmax
assignments, but their exoticity belongs to the event structure and
the admissible-weight polytope, not to the exponential function.

The resulting lesson is simple: local normalization,
single-valuedness/no-disturbance, Cauchy/Gleason linearity, and
physical or cognitive realizability are distinct. Softmax addresses
the first. The pasted logic imposes the second. Cauchy-type arguments
may address the third on richer domains. Classical, quantum, and
exotic probability theories concern the fourth.

\begin{acknowledgments}
This research was funded in whole or in part by the Austrian Science
Fund (FWF), Grant DOI:10.55776/PIN5424624. The author acknowledges
TU Wien Bibliothek for financial support through its Open Access
Funding Programme.
\end{acknowledgments}

\bibliography{svozil}

\end{document}